%% file: main.tex
\documentclass[runningheads, a4paper]{llncs}

\usepackage[ngerman,main=english]{babel}
\usepackage{amsmath}
\usepackage{multicol}
\usepackage{hyperref}
\usepackage{float}
\usepackage{xparse}
\usepackage{ifthen}

\usepackage{graphicx}
\usepackage{subcaption}

\usepackage[ruled, algonl, vlined, linesnumbered]{algorithm2e}
\usepackage[dvipsnames]{xcolor}

\newcommand{\comment}[1]{}

\NewDocumentCommand\Win{gg}{% Two optional arguments
    \textit{Win}\ensuremath{^{\IfNoValueTF{#1}{}{\textit{#1}}}_{\IfNoValueTF{#2}{}{\textit{#2}}}}\xspace
}
\NewDocumentCommand\width{g}{% One optional argument
    \textit{width}\ensuremath{_{\IfNoValueTF{#1}{}{#1}}}\xspace
}
\NewDocumentCommand\row{g}{% One optional argument
    \textit{row}\ensuremath{_{\IfNoValueTF{#1}{}{#1}}}\xspace
}

\newcommand{\depth}{\textit{depth}\xspace}
\newcommand{\shift}{\textit{shift}\xspace}
\newcommand{\lateral}{\textit{Lateral}\xspace}

\newcommand{\floor}[1]{\left\lfloor #1 \right\rfloor}

\newboolean{showcomments}
\setboolean{showcomments}{false}

\newcommand{\ptcom}[1]{\ifthenelse{\boolean{showcomments}}{\textcolor{cyan}{[\textit{PT: #1}]}}{}}
\newcommand{\kvgcom}[1]{\ifthenelse{\boolean{showcomments}}{\textcolor{violet}{[\textit{KvG: #1}]}}{}}
\newcommand{\mpcom}[1]{\ifthenelse{\boolean{showcomments}}{\textcolor{magenta}{[\textit{MP: #1}]}}{}}

\title{How to Relax Instantly: Elastic Relaxation of Concurrent Data Structures}
\titlerunning{Elastic Relaxation of Concurrent Data Structures}

\author{
Kåre von Geijer\inst{1}\orcidID{0009-0007-4823-6855} \and
Philippas Tsigas\inst{1}\orcidID{0000-0001-9635-9154}}
\authorrunning{K. von Geijer and P. Tsigas}

\institute{Chalmers University of Technology, Gothenburg, Sweden}

\begin{document}

\maketitle

\input{Sections/Abstract}

%\comment{
    \input{Sections/Introduction/index}

    \input{Sections/Related-work}

    \input{Sections/Algorithms/index}

    \input{Sections/Correctness/Correctness}

    \input{Sections/Evaluation/index}

    \input{Sections/Conclusion}

    \bibliography{refs}
    
    \newpage
    \appendix

    \input{Sections/Appendix/2Dc-stack-proof}

%}

\end{document}

%% file: Sections/Abstract.tex
\begin{abstract}

% Introduce the use case for relaxed semantics, and that dynamic relaxation is needed
The sequential semantics of many concurrent data structures, such as stacks and queues, inevitably lead to memory contention in parallel environments, thus limiting scalability. Semantic relaxation has the potential to address this issue, increasing the parallelism at the expense of weakened semantics. Although prior research has shown that improved performance can be attained by relaxing concurrent data structure semantics, there is no one-size-fits-all relaxation that adequately addresses the varying needs of dynamic executions.

In this paper, we first introduce the concept of \textit{elastic relaxation} and consequently present the \lateral structure, which is an algorithmic component capable of supporting the design of elastically relaxed concurrent data structures. Using the \lateral, we design novel elastically relaxed, lock-free queues and stacks capable of reconfiguring relaxation during run time. We establish linearizability and define upper bounds for relaxation errors in our designs. Experimental evaluations show that our elastic designs hold up against state-of-the-art statically relaxed designs, while also swiftly managing trade-offs between relaxation and operational latency. We also outline how to use the \lateral to design elastically relaxed lock-free counters and deques.

% Should be in here for Springer templ
\keywords{concurrent data structures \and lock-free \and relaxed semantics}

\end{abstract}

%% file: Sections/Introduction/index.tex
\section{Introduction}\label{sec:introduction}

% Multicores are widespread, and we need effective ways to utilize them. Relaxation is one such proposed way.
As hardware parallelism advances with the development of multicore and multiprocessor systems, developers face the challenge of designing data structures that efficiently utilize these resources. Numerous concurrent data structures exist\,\cite{the-art-of-mpp}, but theoretical results such as \cite{inherent-seq-of-concurrent-obj} demonstrate that many common data structures, such as queues, have inherent scalability limitations as threads must contend for a few access points. 
One of the most promising solutions to tackle this scalability issue is to relax the sequential specification of data structures\,\cite{data-structures-in-multicore-age}, which permits designs that increase the number of memory access points, at the expense of weakened sequential semantics.

% Introduce the out-of-order relaxation, as well as the concept of rank error
% Now we mainly focus on bounded out of order here, which is good (we can save randomized relaxation until later)
The $k$ \textit{out-of-order} relaxation formalized in \cite{quantitative-relaxation} is a popular model\,\cite{2D,distributed-queues,k-queue,klsm} that allows relaxed operations to deviate from the sequential order by up to $k$; for example, for the dequeue operation on a FIFO queue, any of the first $k+1$ items can be returned instead of just the head. This error, the distance from the head for a dequeue, is called the \textit{rank error}. \kvgcom{Can we remove this? We already have two sentences on access points above:} The advantage of relaxation is that multiple items can be inserted and removed at once, reducing contention at a few select access points. Furthermore, while other relaxations, such as quiescent consistency\,\cite{quiescent-consistency} are incompatible with linearizability\,\cite{linearizability}, k-out-of-order relaxation can easily be combined with linearizability, as it modifies the semantics of the data structure instead of the consistency. Despite extensive work on out-of-order relaxation\,\cite{distributed-queues,quantitative-relaxation,2D,family-relaxed-concurrent-queues,multiqueues-new,multiqueues-can-be-priority-scheduler,klsm}, almost all existing methods are static, requiring a fixed relaxation degree during the data structures' lifetime.

% Introducing elastic relaxation, our research problem in this paper
In applications with dynamic workloads, such as bursts of activity with throughput constraints, it is essential to be able to temporarily sacrifice sequential semantics for improved performance in times of high contention. This is the problem tackled in this paper, to specify and design relaxed data structures where the relaxation is reconfigurable during run-time, which we term \textit{elastic relaxation}. Elastically relaxed data structures enable the design of instance-optimizing systems, an area that is evolving extremely rapidly across various communities\,\cite{kraska2021towards}. The trade-off between rank error and throughput is highlighted in \cite{multiqueues-new} and \cite{multiqueues-can-be-priority-scheduler}, where their shortest-path benchmarks show that increased relaxation leads to higher throughput, but at the expense of additional required computation. 

% Sub-structures and static 2D framework
Several relaxed data structures are implemented by splitting the original concurrent data structure into disjoint \textit{sub-structures}, and then using load-balancing algorithms to direct different operations to different sub-structures. In this paper, we base our elastic designs on the relaxed 2D framework presented in \cite{2D}, which has excellent scaling with both threads and relaxation, as well as provable rank error bounds. The key idea of the 2D framework is to superimpose a window (\Win) over the sub-structures, as seen in green in Figure \ref{fig:2d-queue-static} for the 2D queue, where operations inside the window can linearize out of order. The \Win{tail} shifts upward by \depth when it is full, and \Win{head} shifts upward when emptied, to allow further operations. The size of the window dictates the rank error, as a larger one allows for more reorderings.

\begin{figure}
    \centering
    \begin{minipage}[t]{0.4\textwidth}
        \centering
        \includegraphics[height=120px]{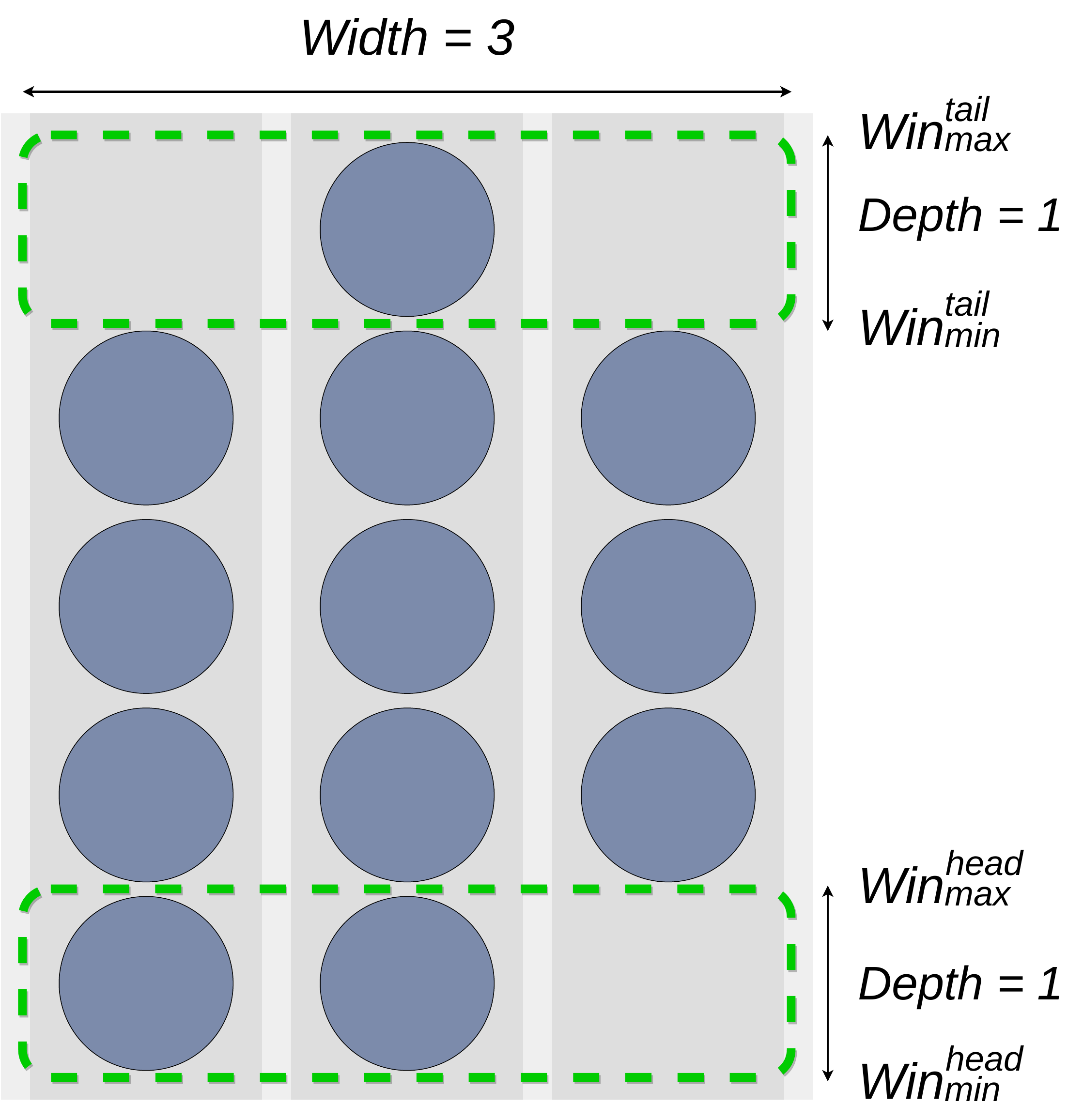}
        \caption{The 2D queue has two windows defining the operable area for the enqueue and dequeue operations.}
        \label{fig:2d-queue-static}
    \end{minipage}
    \hfill
    \begin{minipage}[t]{0.55\textwidth}
        \centering
        \includegraphics[height=120px]{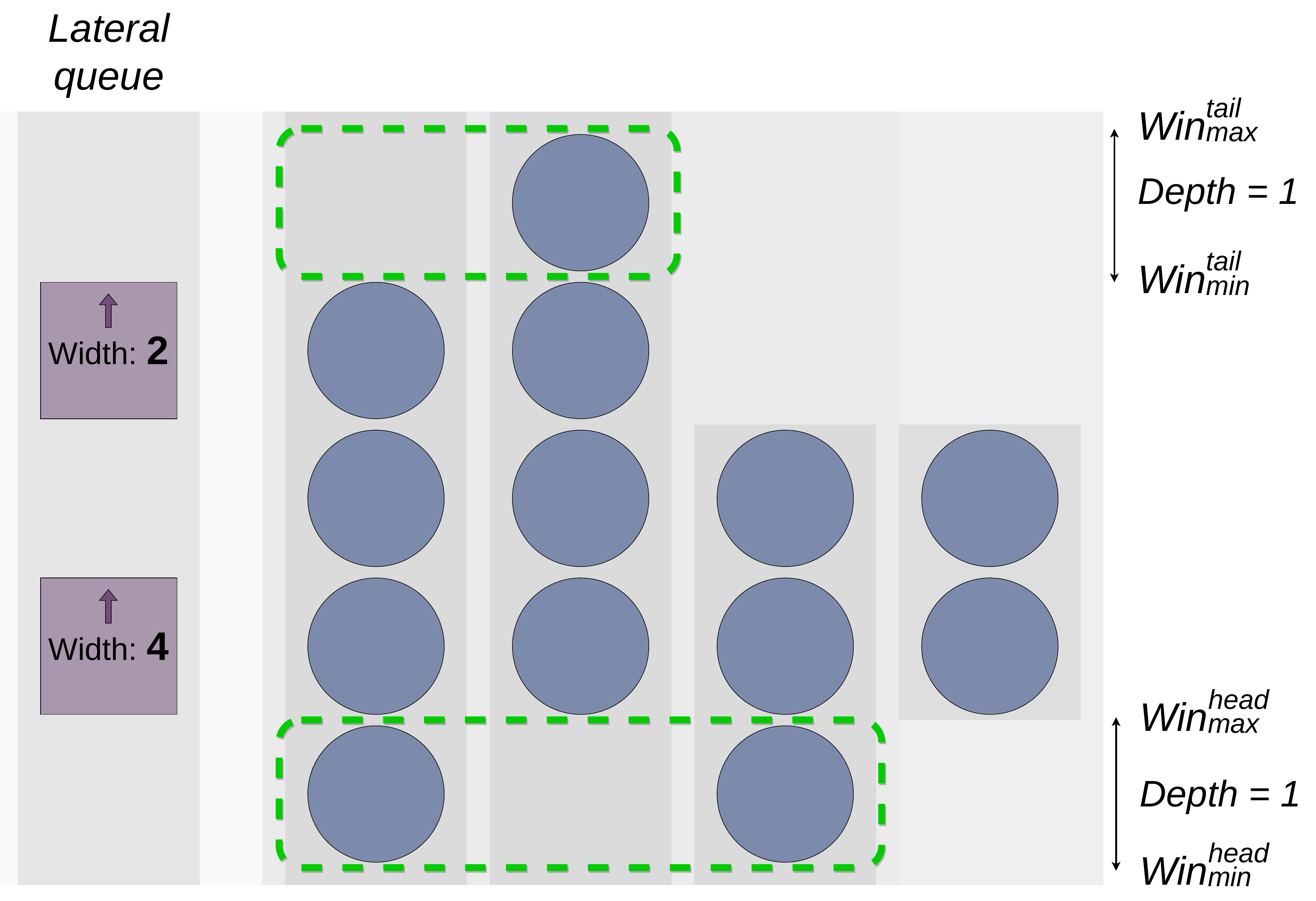}
        \caption{By adding a \lateral to the 2D queue, changes in width at \Win{tail} can be tracked and adjusted to by \Win{head}.}
        \label{fig:2d-queue-lateral}
    \end{minipage}
\end{figure}

% How do we extend the 2D with our lateral?
The algorithmic design concept we propose in this paper is the \lateral structure that can extend the 2D structures to encompass elastic relaxation. The \lateral is added as a side structure next to the set of sub-structures in order to keep track of the elastic changes, as shown in Figure \ref{fig:2d-queue-lateral}. We show how to incorporate the \lateral into the window mechanism that the 2D framework introduced while achieving a deterministic rank error bound. 
Although we chose to use the 2D framework as a base for our designs, the \lateral can also accommodate other designs, such as the distributed queues from \cite{distributed-queues}, the k-queue from \cite{k-queue}, and the k-stack from \cite{quantitative-relaxation}.

% List of contributions.
\vspace{-1.3em}
\subsubsection{Contributions.}
This work takes crucial steps toward designing reconfigurable relaxed concurrent data structures with deterministic error bounds, capable of adjusting relaxation levels during run-time.
% \vspace{-2pt}
\begin{itemize}

   \item Firstly, we introduce the concept of \textit{elastic relaxation}, where the rank error bound can change over time. In addition, we introduce a generic algorithmic extension, the \lateral, to efficiently extend many relaxed data structure designs to support elasticity.
    \item We design and implement two elastically relaxed queues and a stack based on the \lateral, as well as outline how to incorporate it into other data structures such as counters and deques.
    \item We establish rank error bounds for our elastic designs, as well as prove their correctness.
    \item We present an extensive evaluation of our proposed data structures, benchmarking against non-relaxed data structures and statically relaxed ones. These evaluations show that the elastic designs significantly outscale non-relaxed data structures and perform as well as the statically relaxed ones, while also supporting elastic relaxation.
    \item Finally, we showcase the elastic capabilities of our design by implementing a simple controller for a producer consumer application. By dynamically adjusting the relaxation, it is able to control the producer latency during bursts of activity with minimal overhead.
\end{itemize}

% Describe structure of the paper
\vspace{-1.5em}
\subsubsection{Structure.}
Section \ref{sec:related work} introduces related work, as well as gives a short description of the structures on which we base our elastic designs. Section \ref{sec:algorithms} introduces elastic relaxation and our novel data structures, which we then prove correct and provide worst-case bounds for in Section \ref{sec:correctness}. Section \ref{sec:evaluation} experimentally evaluates the new algorithms, both comparing them to earlier non-elastic data structures as well as testing their elastic capabilities. Section \ref{sec:conclusion} contains a few closing remarks.

%% file: Sections/Related-work.tex
\section{Related Work}\label{sec:related work}

% Historical background, relaxation without calling it relaxation
One of the earliest uses of relaxed data structures is from 1993 by Karp and Zhang in \cite{karp-zhang-original} where they used process-local work queues, only occasionally synchronizing to redistribute items randomly. 
These queues functioned as heuristics in their branch-and-bound algorithm to dictate the order of exploration in the problem space — which to this day is a prevalent application of relaxation as it does not impede correctness\,\cite{multiqueues-can-be-priority-scheduler,lightweight-graph-analytics,distributed-queues}. Although their paper used the idea of relaxation to improve performance, the concept of semantic relaxation was not defined, introduced, or properly explored.

% Modern relaxed data structures, loads of links, and stating that they are efficient.
More recently, relaxed data structures have emerged as a promising technique to boost concurrency\,\cite{data-structures-in-multicore-age} and are formally analyzed in for example \cite{quantitative-relaxation,distributional-linearizability,upper-and-lower-bounds,power-of-choice}. They have demonstrated exceptional throughput on highly parallel benchmarks \cite{2D,multiqueues-new,scal,distributed-queues}, and have shown to be suitable in heuristics for e.g. graph algorithms\,\cite{lightweight-graph-analytics,multiqueues-can-be-priority-scheduler}.

% Quantitative relaxation
Henzinger et al. specified \textit{quantitative relaxation} in \cite{quantitative-relaxation} as a framework for defining relaxed data structures with a hard bounded error. These relaxations, such as $k$ out-of-order are easily extendably to elastic relaxation by letting the bound change during run-time. The paper also introduces the relaxed $k$-segment stack, which in turn builds upon the earlier relaxed FIFO queue from \cite{k-queue}. The size of these segments ($k$) is similar to the \width dimension in the 2D framework\,\cite{2D}, and it can therefore be extended to be elastic with our \lateral.

% Randomized relaxation
Another direction of relaxed data structures forgoes the hard bound from \cite{quantitative-relaxation} and instead relies on randomness to only give probabilistic guarantees, and was formalized in \cite{distributional-linearizability} as an extension to \cite{quantitative-relaxation}. A successful design of these probabilistically relaxed data structures is the MultiQueue\,\cite{multiqueues-short,multiqueues-new} which builds upon The Choice of Two\,\cite{power-of-two}. It is a relaxed priority queue which enqueues items into random sub-queues, while dequeuing the highest priority item from a random selection of two sub-queues. This can similarly be applied to FIFO queues such as the d-RA in \cite{distributed-queues} or stacks. While not providing upper bounds on relaxation, the relaxation of the MultiQueue has been extensively analysed in \cite{power-of-choice}.

The SprayList from \cite{spraylist} is another relaxed priority queue, but does not use multiple sub-queues. Instead, it builds upon a concurrent skip-list and does a random walk over items at each delete-min, returning one of the top $\mathcal{O}(p\log^3p)$ items with high probability, where $p$ is the number of threads. The properties of this random walk can be adjusted at run-time, essentially making it \textit{elastically relaxed}, although with a probabilistic rank error guarantee instead of a bound. However, experiments from \cite{multiqueues-new,multiqueues-short,multiqueues-can-be-priority-scheduler} suggest that the SprayList is outperformed by the MultiQueue, for which there is no known efficient elastic design.

% Low latency queues
Many concurrent queues aim for high throughput, yet the RCQs, detailed in \cite{family-relaxed-concurrent-queues}, targets reduced wait-time by leveraging the LCRQ framework\,\cite{lcrq}. It achieves similar throughput to fetch-and-add methods but with significantly lower wait-times, especially in oversubscribed, nearly empty queues. While offering a new perspective on concurrent queue design, the RCQ lacks a mechanism to adjust its relaxation feature.

\subsubsection{2D Framework}

% Introduce multiple sub-structures and width
The 2D framework for $k$ out-of-order relaxed data structures from \cite{2D} (hereby called the \textit{static} 2D designs) outscales other implementations from the literature such as \cite{k-queue,quantitative-relaxation,distributed-queues} with threads, while unifying designs across stacks, queues, deques, and counters. Furthermore, its throughput scales monotonically with $k$. It achieves this, as shown in Figure \ref{fig:2d-queue-static}, by superimposing a \textit{window} (\Win) over multiple disjoint concurrent (nonrelaxed) sub-structures which defines which of them are valid to operate on. %This leads to high scalability as the window relatively infrequently changes.

%We chose the 2D framework for \textit{k out-of-order} relaxed data structures from \cite{2D} (from here on called the \textit{static} 2D designs) as the base for our new designs as it (i) unifies design across various data structures, (ii) scales monotonically with increasing relaxation, (iii) provides deterministic worst-case guarantees, and (iv) achives notable improvements compared to preceding works in the literature. 

% Introduce depth and the figure. Should we have some more 
The 2D \textit{window} has a \width (always the number of sub-structures in the static designs) and a \depth, which together govern the relaxation profile. At any point, it is valid to insert an item on a \textit{row} $r$ where $r \leq \Win{}{max}$, or delete an item on row $r$ where $r > \Win{}{min} \equiv \Win{}{max} - \depth$, while keeping linearization order in the sub-structures. To maximize data locality, each thread tries to do as many operations as possible on the same sub-structure in a row, only moving due to contention or from reaching \Win{}{max} or \Win{}{min}.

% Shifting the window, cache locality
If an operation cannot find a valid sub-structure, it will try to \textit{shift} the window. For example, if a thread tries to insert an item into the 2D queue and sees all sub-queues at \Win{tail}{max}, it will try to shift the window up by atomically incrementing \Win{tail}{max} (and implicitly \Win{tail}{min}) by $\depth$, after which it restarts the search for a valid sub-queue.

% Difference between the stack and queue
The 2D stack is similar to the 2D queue, but as both insert (push) and delete (pop) operate on the same side of the data structure, only one window is used. This leads to \Win{}{max} no longer increasing monotonically, but instead decreasing under heavy delete workloads, and increasing under insert heavy ones. Furthermore, \Win{}{max} shifts by $\depth/2$ instead of $\depth$, which approximately makes it fair for future push and pop operations. The framework also covers deques and counters, using the same idea of a window to define valid operations.

%% file: Sections/Algorithms/index.tex
\section{Design of Elastic Algorithms}\label{sec:algorithms}

% Introducing elastic relaxation
Static $k$ out-of-order relaxation is formalized in \cite{quantitative-relaxation} by defining and bounding a \textit{transition cost} (rank error) of the “get” methods within the linearized concurrent history. Elastic relaxation allows the relaxation configuration to change over time, which will naturally change the bound $k$ as well. Therefore, we define \textit{elastic out-of-order} data structures as static out-of-order, but allow the rank error bound to be a function of the relaxation configuration history during the lifetime of the accessed item. In the simplest case, such as the elastic queue from Section \ref{sec:law-queue}, the rank error bound for every dequeued item is a function of the \width and \depth during which the item was enqueued (which is the same as when it is dequeued).

% Elastic 2D = depth and width elastic
To elastically extend the 2D algorithms, we want the \textit{width} and \textit{depth} of the 2D windows be adjustable during run-time, which would enable fine-grained control of the relaxation profile. Our designs let these parameters change every window shift, and then update them atomically with \Win{max}. Changing \depth is practically simple by including it as a window variable \Win{}{depth} (although it has implications on the error bounds). Changing the \width efficiently requires more attention. We create a maximum number of sub-structures at initialization time (each taking up only a couple bytes of memory), and then use the \lateral structure to keep track of which sub-structures have been inserted into for which rows as seen in Figure \ref{fig:2d-queue-lateral}.

% Definition of a lateral
\begin{definition}
    A \lateral to a relaxed data structure is a set of nonoverlapping adjacent ranges of rows, where each range has a corresponding width bound.
\end{definition}

Furthermore, we call the \lateral \textit{consistent} if the \textit{width bound} of each node properly bounds the width of the corresponding rows in the main structure. The exact implementation of the \lateral will vary depending on what performance properties are desired, but the overarching challenge is to keep it consistent while also being fast to read and update, as well as promising good rank error bounds. For our 2D designs, we found that it is essential to only update the \lateral (at most) once every window. % By correctly reflecting the row widths at window shifts, it can be used to set the \width to delete items from in the next window.

\input{Sections/Algorithms/windowless-queue}
\input{Sections/Algorithms/elastic_queue}
\input{Sections/Algorithms/elastic_stack}

\input{Sections/Algorithms/further-extentions}

%% file: Sections/Algorithms/windowless-queue.tex
\subsection{Elastic Lateral-as-Window 2D Queue} \label{sec:law-queue}

% Describe the merge of the window and lateral
This first elastic Lateral-as-Window queue (LaW queue) merges the window into the \lateral, and can be applied to most data strucutures from the 2D framework\,\cite{2D} with small changes. The pseudocode is shown in Algorithm \ref{alg:law-queue}. First, we add a \lateral queue that is implemented as a Michael-Scott queue\,\cite{ms-queue}, for which the code omits the standard failable \textit{Enqueue} and \textit{Dequeue} methods. The \lateral nodes are windows, where each window contains \Win{}{max}, \Win{}{depth}, and \Win{}{width}. The \Win{tail} and \Win{head} then become the head and tail nodes in the \lateral. Every shift of \Win{tail} enqueues a new window in the \lateral (line \ref{alg:law-queue-pshift}), and every shift of \Win{head} dequeues a window (line \ref{alg:law-queue-gshift}).

% Pseudo code for the windowless lateral
\begin{algorithm}[!ht]
\SetKwProg{Func}{function}{}{}
\SetKwProg{Method}{method}{}{}
\SetKwProg{Struct}{struct}{}{}
\SetKwComment{NormComment}{// }{}

\caption{Pseudocode for the \lateral in the elastic LaW queue}\label{alg:law-queue}
\scriptsize
%\footnotesize
%\DontPrintSemicolon
%\SetKwComment{Comment}{$\triangleright$\,}{}

\begin{multicols}{2}

\Struct{Window}{
    Window* next\; 
    uint max\;
    uint depth\;
    uint width\;
}

global \Struct{Lateral}{
    Window* head\;
    Window* tail\;
}

\NormComment{Try to atomically enqueue \textit{new} directly after \textit{expected}}
\textbf{method} \textit{Lateral.Enqueue(expected, new)}\;

\NormComment{Try to atomically dequeue \textit{expected} if it is the head}
\textbf{method} \textit{Lateral.Dequeue(expected)}\;

\Method{Lateral.ShiftTail(old\_window) \label{alg:law-queue-pshift}}{
    depth $\gets$ depth$_{\text{shared}}$\;
    \DontPrintSemicolon
    new\_window $\gets$ \{ \;
      \hspace{1em}width: width$_{\text{shared}}$, \; \label{alg:law-queue-width}
      \hspace{1em}depth: depth, \;
      \hspace{1em}max: old\_window.max + depth\;
    \PrintSemicolon
    \}\;
    Lateral.Enqueue(old\_window, new\_window)\;
}

\Func{ShiftHead(current\_head) \label{alg:law-queue-gshift}}{
    Lateral.Dequeue(current\_head)\;
}

\end{multicols}
\vspace{1em}
\end{algorithm}

% Introducing the elasticity, and covering problems
As shown in \textit{ShiftTail} (line \ref{alg:law-queue-pshift}), \Win{}{depth} and \Win{}{width} can be updated from global variables every shift, which enables the elasticity. The main drawback of this design is that the relaxation only changes at the tail, and must propagate through the queue to reach the head. The \lateral also takes up more memory than two simple global windows, and require a bit more overhead to update than shifting the windows with a simple compare-and-swap.

% Other reconciliations
Other than the pseudocode in Algorithm \ref{alg:law-queue}, the remaining logic from the static 2D queue require only small adjustments. Mainly, items must always be inserted within the window, so each items gets enqueued at $\textit{max}(\Win{put}{max} - \Win{put}{depth}, \textit{last\_item.row}) + 1$, which create needed gaps in sub-queues as seen in Figure \ref{fig:2d-queue-lateral}. Furthermore, \Win{head} cannot pass \Win{tail}, and dequeues can simply return empty if the \lateral is empty.

%% file: Sections/Algorithms/elastic_queue.tex
\subsection{Elastic Lateral-plus-Window 2D Queue}

% Introduce the name and why it is "better" than the LaW queue.
The elastic Lateral-plus-Window 2D queue (elastic LpW queue) solves two shortcomings of the previous elastic LaW queue. Firstly, it allows the head to change relaxation independently of the tail by letting both windows change \Win{}{depth} at window shifts. Second, it does not have to allocate a new \lateral node every \Win{tail} shift, and instead only creates \lateral nodes when \Win{tail}{width} changes. However, it comes at the expense of having to decouple the window and \lateral components, and uses a 128-bit shared atomic struct for each window.

% Introduce pseudo code for the lateral and window, mention sub-structure search, and lateral as ms-queue
The pseudocode for the \lateral and windows is presented in Algorithm \ref{alg:window-plus-lateral-queue} and shows that the \Win{head} and \Win{tail} structs now are global variables, both containing \textit{max}, \textit{depth}, and \textit{width}. The \lateral is again implemented as a Michael-Scott queue\,\cite{ms-queue} where we omit the \textit{Enqueue} and \textit{Dequeue} implementations.

%%% PSEUDOCODE %%%
% Pseudo code for the window plus lateral queue
\begin{algorithm}[!ht]
\SetKwProg{Func}{function}{}{}
\SetKwProg{Method}{method}{}{}
\SetKwProg{Struct}{struct}{}{}
\SetKwComment{NormComment}{// }{}

\caption{Lateral and window code for the elastic LpW queue}\label{alg:window-plus-lateral-queue}
\scriptsize
%\footnotesize
%\DontPrintSemicolon
%\SetKwComment{Comment}{$\triangleright$\,}{}

\begin{multicols}{2}

global \Struct{TailWindow}{
    uint64 max\;
    uint16 depth\;
    uint16 width\;
    uint16 next\_width\;
}

global \Struct{HeadWindow}{
    uint64 max\;
    uint16 depth\;
    uint16 width\;
}

\Struct{LateralNode}{
    LateralNode* next\;
    uint row\;
    uint width\;
}

global \Struct{Lateral}{
    LateralNode* head\;
    LateralNode* tail\;
}

\NormComment{Try to atomically enqueue \textit{new} directly after \textit{expected}}
\textbf{method} \textit{Lateral.Enqueue(expected, new)}\;

\NormComment{Try to atomically dequeue \textit{expected} if it is the head}
\textbf{method} \textit{Lateral.Dequeue(expected)}\;

\Method{Lateral.SyncTail(window) \label{alg:lpw-queue-sync-tail}}{
    tail $\gets$ Lateral.tail\;
    \If{tail.row $\leq$ window.max \label{alg:lpw-queue-lat-check}}{
        \DontPrintSemicolon
        new\_tail $\gets$ \{ \;
          \hspace{1em}row: window.max + 1, \; \label{alg:lpw-queue-lat-max}
          \hspace{1em}width: window.next\_width \;
        \PrintSemicolon
        \}\;
        Lateral.Enqueue(tail, new\_tail)\;
    }
}

\Func{ShiftTail(old\_window) \label{alg:lpw-queue-shift-tail}}{
    \textbf{if} window.width $\neq$ window.next\_width Lateral.SyncTail(old\_window) \label{alg:lpw-diff-width}\;

    depth $\gets$ depth$_\text{shared}$\;
    \DontPrintSemicolon
    new\_window $\gets$ \{ \;
      \hspace{1em}width: old\_window.next\_width, \;
      \hspace{1em}next\_width: width$_\text{shared}$, \; \label{alg:lpw-queue-width-shared}
      \hspace{1em}depth: depth, \;
      \hspace{1em}max: old\_window.max + depth \;
    \PrintSemicolon
    \}\;

    CAS(\&TailWindow, old\_window, new\_window)\;
}

\Func{ShiftHead(old\_window) \label{alg:lpw-queue-shift-head}}{

    % Dequeue all nodes behind the current window
    \While{true} {
        head $\gets$ Lateral.head()\;
        \textbf{if} head.max $>$ old\_window.max \textbf{break}\; 
        Lateral.Dequeue(head)\;
    }

    % Optimistic window definition
    \DontPrintSemicolon
    new\_window $\gets$ \{ \;
      \hspace{1em}width: old\_window.width, \;
      \hspace{1em}max: max(TailWindow.max, \;\hspace{4em}old\_window.max + depth$_\text{shared}$) \label{alg:lpw-head-overtake}\;
    \PrintSemicolon
    \}\;

    \If{head.row = old\_window.max + 1 \label{alg:lpw-overlapping-lateral}}{
        new\_window.width $\gets$ head.width\;  \label{alg:lpw-new-width}
        head $\gets$ head.next\;
    } 
    \If {head.row $<$ new\_window.max \label{alg:lpw-lateral-limit-max}} {
        new\_window.max $\gets$ head.row - 1\;  
    }
    new\_window.depth $\gets$ new\_window.max - old\_window.max\;

    CAS(\&TailWindow, old\_window, new\_window)\;
}

\end{multicols}
\vspace{1em}
\end{algorithm}

% Shifting the tail
Shifting \Win{tail} is done during a pending enqueue call when the window has become full (line \ref{alg:lpw-queue-shift-tail}). It reads the desired \textit{depth} and \textit{width} from shared variables, and use them in the next window. However, the \textit{width} is not used immediately, but instead written to a \textit{next\_width} field in the future window, which is then used in the successive shift as the new \textit{width}. This delay is introduced to ensure that a \lateral node will be enqueued with the new \textit{width} before this \textit{width} is used in an enqueue. Enqueueing the \lateral node is done before the window shift when $\textit{next\_width} \neq \textit{width}$ (line \ref{alg:lpw-diff-width}), ensuring that the head of the queue will be aware of changes in width before they occur.

% Shifting the head
Similarly, \Win{head} is shifted during a dequeue call when all sub-queues in the \textit{width} has reached the window \textit{max}. The shift starts by dequeueing all \lateral nodes below the current \Win{head}{max}, as they represent stale changes in width. The shift optimistically reads a shared \textit{depth} variable which is used for the next window (line \ref{alg:lpw-head-overtake}), unless the \lateral head node overlaps the new range (line \ref{alg:lpw-lateral-limit-max}). If the \lateral head is on the bottom row of the new window (line \ref{alg:lpw-new-width}), then the new window will adapt to the width change encoded by the head. Otherwise, the new \Win{head}{max} is limited to not overlap a \lateral node (line \ref{alg:lpw-lateral-limit-max}). This ensures that all nodes within a window were pushed within the same \textit{width}, which is used at dequeues to calculate which row the dequeued node was at.

% Summary and bit counting
At the cost of separating the \lateral and the window, this LpW queue is able to change \textit{depth} independently for the head and the tail. However, the \textit{width} is still only ever changed at \Win{tail} which \Win{head} has to adapt to by using the \lateral. We have designed this to be efficient on modern x86-64 machines where CAS only has hardware support for up to 128 bits, which then becomes the upper size limits for our window structs. One can allocate the sizes differently depending on the need of the application, but if 128 bits is not enough, or a machine without 128-bit CAS support is used, the elastic LaW-queue might be more suitable.

%% file: Sections/Algorithms/elastic_stack.tex
\subsection{Elastic Lateral-plus-Window 2D Stack}

This elastic Lateral-plus-Window 2D stack (elastic LpW stack) uses the same methods as the elastic LpW queue to make the static 2D stack elastic. The pseudocode for the \lateral and window is found in Algorithm \ref{alg:lpw-stack} and uses a global shared window struct, which is updated with CAS at window shifts (line \ref{alg:lpw-stack-shift-cas}), and a Treiber stack\,\cite{treiber-stack} as a \lateral for all changes in \textit{width}. The nonmonotonic nature of the window means that the width bound of a row can change many times, requiring the \lateral to be stabilized before every window shift (line \ref{alg:lpw-stack-stabilize}).

%%% PSEUDOCODE %%%
% Pseudo code for the elastic stack
\begin{algorithm}[!ht]
\SetKwProg{Func}{function}{}{}
\SetKwProg{Method}{method}{}{}
\SetKwProg{Struct}{struct}{}{}
\SetKwComment{Comment}{$\triangleright$\,}{}
\SetKwComment{NormComment}{//\,}{}

\caption{Lateral and window code for the elastic LpW stack}\label{alg:lpw-stack}
\fontsize{6pt}{7.5pt}\selectfont % This is between scriptsize and tiny
%\scriptsize
%\footnotesize
%\DontPrintSemicolon

\begin{multicols}{2}

\NormComment{Example size of fields to fit 128 bits}
global \Struct{Window}{
    uint32 max\;
    uint16 depth\;
    uint16 push\_width\;
    uint16 pop\_width\;
    uint16 last\_push\_width\;
    enum last\_shift\Comment*[r]{UP or DOWN\hspace{2em}}
    uint31 version\;
}

\Method{Window.Min(self)}
{
    \Return self.max - self.depth\;
}

\Func{Shift(dir, old\_window) \label{alg:lpw-stack-shift}}{
    Lateral.Stabilize(old\_window)\; \label{alg:lpw-stack-stabilize}

    \DontPrintSemicolon
    new\_window $\gets$ \{ \;
      \hspace{1em}push\_width: width$_\text{shared}$,\; \label{alg:lpw-stack-push-width}
      \hspace{1em}last\_push\_width: old\_window.push\_width,\;
      \hspace{1em}depth: depth$_\text{shared}$,\;
      \hspace{1em}last\_shift: dir,\;
      \hspace{1em}version: old\_window.version + 1\;
    \PrintSemicolon
    \}\;

    \eIf{dir $=$ UP}
    {
        new\_window.max $\gets$ old\_window.max + new\_window.depth$/2$\; \label{alg:lpw-stack-shift-up-row}
    }
    {
        new\_window.max $\gets$ old\_window.Min() + new\_window.depth$/2$\;\label{alg:lpw-stack-shift-down-row}
    }

    % To we even need this? This basically a detail    
    %new\_window.max $\gets$ max(new\_window.depth, new\_window.max)\; \label{alg:stack-shift-max-limit}
    new\_window.pop\_width $\gets$ max(new\_window.push\_width, Lateral.width(new\_window.Min()))\; \label{alg:lpw-stack-pop-width}

    CAS(\&Window, old\_window, new\_window)\; \label{alg:lpw-stack-shift-cas}
}

global \Struct{Lateral}{
    LateralNode* top\;
    uint64 version\;
}

\Struct{LateralNode}{
    LateralNode* next\;
    uint row\;
    uint width\;
}

\Method{Lateral.Width(self, row)}
{
    \Return max(node.width \textbf{for} node \textbf{in} self \textbf{where} node.row $>$ row)\; \label{alg:pop_width}
}

% Is it clear enough that we are working on a local copy?
\Method{Lateral.Stabilize(self, win) \label{alg:sync-lat-decl}}{
    read\_lat $\gets$ new\_lat $\gets$ self.read()\;
    \If{Window $\neq$ win \textbf{or} read\_lat.version $=$ win.version}
    {
        \Return\Comment*[r]{Already completed\hspace{0em}}
    }
    \NormComment{Phase 1: Update local \lateral} \label{alg:lpw-stack-lat-phase1}
    new\_lat.top $\gets$ read\_lat.top.Update(win)\; 

    \NormComment{Phase 2: Push new top if changed width} \label{alg:lpw-stack-lat-phase2}
    upper\_bound $\gets$ max(win.max, stack.row \textbf{for} stack \textbf{in} substacks[win.push\_width..win.last\_push\_width])\; \label{alg:lat-upper-bound-calc}
    \uIf{win.push\_width $>$ win.last\_push\_width and new\_lat.top.row $<$ win.Min()}
    {
        \DontPrintSemicolon
        node' $\gets$ \{ \Comment*[r]{New top node\hspace{0em}} \label{alg:lat-smaller-push-row}
            \hspace{1em}row: win.Min(),\;
            \hspace{1em}width: win.last\_push\_width,\;
            \hspace{1em}next: new\_lat.top\;
        \PrintSemicolon
        \}\;
        new\_lat.top $\gets$ node'\Comment*[r]{Push width change\hspace{0em}}
    }
    \uElseIf{win.push\_width $<$ win.last\_push\_width and new\_lat.top.row $<$ upper\_bound}
    {
        \DontPrintSemicolon
        node' $\gets$ \{ \Comment*[r]{New top node\hspace{0em}} \label{alg:lat-upper-bound}
            \hspace{1em}row: upper\_bound,\;
            \hspace{1em}width: win.last\_push\_width,\;
            \hspace{1em}next: new\_lat.top\;
        \PrintSemicolon
        \}\;
        %node'.row $\gets$ upper\_bound\Comment*[r]{New top node\hspace{0em}} \label{alg:lat-upper-bound}
        %node'.width $\gets$ win.last\_push\_width\;
        %node'.next $\gets$ new\_lat.top\;
        new\_lat.top $\gets$ node'\Comment*[r]{Push width change\hspace{0em}}
    }

    \If{new\_lat $\neq$ read\_lat} {
        new\_lat.version $\gets$ win.version\;
        CAS(\&Lateral, read\_lat, new\_lat)\; \label{alg:cas-lateral-stack}
    }
}

\Method{LateralNode.Update(self, win)}
{
    \If{self.row $\leq$ win.Min() \label{alg:lateral-update-base}}
    {
        \Return self\;
    }
    node' $\gets$ node.Clone()\;%\Comment*[r]{New LatNode\hspace{0em}}
    
    \uIf{node.width $\leq$ win.push\_width}
    {
        % Smaller, so must lower to bottom of window.
        node'.row $\gets$ win.Min()\; \label{alg:lat-smaller-width-update}
    }
    \ElseIf{node.width $>$ win.last\_push\_width and win.last\_shift $=$ DOWN \label{alg:lat-wider-if}}
    {
        % Shifted down, so mush have reached bottom of last window
        last\_min $\gets$ win.max $-$ win.depth$/2$\;
        node'.row $\gets$ min(node.row, last\_min)\; \label{alg:lat-larger-width-update}
    }
    next $\gets$ node'.next.Update(win)\;
    \eIf{node'.row $\leq$ node.next.row} {
        \Return next\Comment*[r]{Remove node\hspace{0em}}\label{alg:lateral-stack-remove}
    }
    {
        node'.next $\gets$ next\;
        \Return node'\;
    }
}

\end{multicols}
\vspace{1em}
\end{algorithm}

% Push and pop overview
The push and pop operations are very similar to the static 2D stack. The \textit{push} operations use the \textit{push\_width}, which is the desired global width (line \ref{alg:lpw-stack-push-width}), and the \textit{pop} operations now use the \textit{pop\_width}, which is the upper width bound on the rows in the window based on the \lateral (line \ref{alg:lpw-stack-pop-width}). Not shown in the pseudocode, the nodes in the sub-stacks also store the row they are pushed at, which is used at pops to update the descriptor for each sub-stack.

% Window shift
The window shift (line \ref{alg:lpw-stack-shift}) updates these widths, as well as shifting the \Win{}{max} or \Win{}{min} by the new $\Win{}{depth}/2$ (lines \ref{alg:lpw-stack-shift-up-row} and \ref{alg:lpw-stack-shift-down-row}). In addition, it stores the last \textit{push\_width} and the direction of the shift (UP or DOWN) in the window, which are used to keep the \lateral consisent for the next window shift. It linearizes with a 16 byte CAS at line \ref{alg:lpw-stack-shift-cas}.

% Lateral sync, invariant
The core of this data structure is the \lateral, and how it is kept consistent for every shift (line \ref{alg:lpw-stack-stabilize}). This function maintains the stack invariant that for any \lateral node $l$, all rows in the 2D stack between $l.row$ and $l.\textit{next}.\textit{row}$ will have smaller or equal width as $l.\width$. We divide this synchronization into two phases, which together create a local top candidate for the \lateral stack, and linearize with a CAS at line \ref{alg:cas-lateral-stack}.

% First phase, lowering lateral nodes
The first phase clones and tries to \textit{lower} \lateral nodes above \Win{}{min} (line \ref{alg:lpw-stack-lat-phase1}). By lowering a \lateral node $l$ to row $r$, we set $l.\row \gets min(l.\row, r)$, and then if $l.\row \leq l.\textit{next}.\row$, $l$ is removed from the stack (by linking its parent $l', l'.\textit{next} \gets l.\textit{next}$). If a \lateral node $l, l.\row > \Win{}{min}$ has $l.\width \leq \Win{}{push\_width}$, then it is lowered to \Win{}{min} (line \ref{alg:lat-smaller-width-update}), as new nodes can have been pushed outside $l.\width$ within the window, invalidating its bound. Otherwise, if $l.\width > \Win{}{last\_push\_width}$ and the last shift was downwards, all sub-stacks must have been seen at the previous \Win{}{min} before it shifted down, so $l$ is lowered to a conservative estimate of the previous $\Win{}{min}$ (line \ref{alg:lat-larger-width-update}). This estimate only deviates from the actlal previous \Win{}{min} when the last \Win{min} was smaller than $\Win{}{depth}/2$, in which case it overestimates, thus keeping a correct bound.

% Second phase, pushing a change in width
In the second phase (from line \ref{alg:lpw-stack-lat-phase2}), a new \lateral node with \textit{width} \Win{}{last\_push\_width} is pushed if the width has changed.\begin{itemize}
    \item If $\Win{}{push\_width} > \Win{}{last\_push\_width}$, the width has increased, and a \lateral node is pushed at \Win{}{min} to signify that the width from there on is smaller. This is not needed for correctness, but is helpful in limiting the \Win{}{pop\_width} if the width shrinks in the future.
    \item If $\Win{}{push\_width} < \Win{}{last\_push\_width}$, the width has decreased and a new \lateral node needs to be pushed at the highest row containing nodes between \Win{}{push\_width} and \Win{}{last\_push\_width}. This can not reliably be calculated from the present window variables, so we simply iterate over the sub-stacks (line \ref{alg:lat-upper-bound-calc}).
\end{itemize}

% Summary, and search strategy
In summary, the elastic LpW stack uses a similar idea as the elastic LpW queue, but needs to do some extra work to maintain the \lateral invariant. However, unless the workload is very pop-heavy, the \lateral nodes should quickly stabilize and let the \textit{push\_width} and \textit{pop\_width} be equal.

%% file: Sections/Algorithms/further-extentions.tex
\subsection{Elastic Extension Outlines: 2D Counter and Deque} \label{sec:elastic-counter} \label{sec:elastic-deque}

% Elastic counter
The 2D counter can easily be made elastic by adding a \lateral counter. However, as the counter is not a list-based data structure, and the "get" operation is more akin to reading its size, it becomes slightly different. The key is to let the \lateral track the difference between the the sum of all counters within \Win{}{width} and the total count. A simple way to implement this is to add a small delay before changing \Win{width}, as in the LpW queue, and iterate over the counters between the next and current \Win{width}, updating them and the \lateral to get a consistent offset.

% Deriving 2Dc deque
To derive an elsatic 2D deque, we use a similar deque as \cite{2D}, but with only one window at each side of the deque, much like the 2D queue. However, as with the 2D stack, these windows can shift both up and down. 
% Making it elastic
This can be made elastic in a similar fashion to the LaW queue, where a \lateral deque is kept with the sequence of all windows. If the windows shift with \depth rows each time, as in the queue, a similar approach to the k-stack from \cite{quantitative-relaxation} can be used to make sure a window is not removed while non-empty. If the window should shift by $\depth/2$ as the stack, successive windows would overlap, requiring extra care. A solution to this could be to split each window in two, letting the threads operate on the two topmost windows, and still use the confirmation technique from \cite{quantitative-relaxation} for the top \Win when shifting downwards.

%% file: Sections/Correctness/Correctness.tex
\section{Correctness}\label{sec:correctness}

% Introduction to what we prove, extending the theory for relaxation
In this section, we prove the correctness and relaxation bounds for our elastic designs. For simplicity, we only relax non-empty remove operations and assume a double-collect\,\cite{double-collect} approach is used to get linearizable empty returns, as is also done in \cite{distributed-queues}. However, it is possible to extend the arguments to allow relaxed empty returns with the same rank error bound $k$. Furthermore, if no elastic changes are used, our elastic designs have the same error bounds as the static 2D structures: $\depth(\width-1)$ for the queue\,\cite{2D} and $2.5\depth(\width - 1)$ for the stack (see Appendix \ref{app:2Dc-stack-bound}).

% Lock-freedom
Our designs are lock-free as (i) each sub-structure, the windows, and the \lateral are updated in a lock-free manner (by linearizing with a CAS), (ii) the \lateral is updated at most once every window, and (iii) that there cannot be an infinite number of window shifts without progress on any of the sub-structures, as was proved in \cite{2D}.

\input{Sections/Correctness/law-queue}

\input{Sections/Correctness/lpw-queue}

\input{Sections/Correctness/lpw-stack}

%% file: Sections/Correctness/law-queue.tex
\subsection{Elastic LaW Queue}

% Sequence of windows

% Rank error bound for 2D queue
% Now we still only relax it for the dequeues. We could add another theorem for relaxing both operations.
\begin{theorem}\label{thm:law-queue-bound}
    The elastic LaW 2D queue is linearizeable with respect to a FIFO queue with elastic $k$ out-of-order relaxed dequeues, where $k = (\Win{head}{width} - 1)\Win{head}{depth}$.
\end{theorem}

\begin{proof}
    The key observation is that every \Win{tail} must fill all \Win{head}{depth} rows of \Win{head}{width} items each before shifting to the next \Win. These items can be enqueued in any order, except that each sub-queue is totally ordered. Enqueues in different windows are on the other hand correctly ordered, due to the sequential semantics of the \lateral. The \Win{head} uses the same \Win structs as the \Win{tail}, by traversing the \lateral of past \Win{tail}, not shifting past such a window until it has also dequeued all its $\Win{}{width} \times \Win{}{depth}$ items.

    Thus, as the oldest items in the queue always are in \Win{head}, and dequeues only ever dequeue from the \Win{head}, together with the fact that the sub-queues are totally ordered, means that a dequeue can at most skip the first $(\Win{head}{width} - 1)\Win{head}{depth}$ items.
\end{proof}

% Lateness
Additionally, there is another measure of the relaxation error called the \textit{lateness} defined in \cite{quantitative-relaxation}, or \textit{delay} in \cite{multiqueues-new}. For a relaxed queue, the lateness is the maximal number of consecutive dequeue operations that can be carried out without dequeuing the head. Due to the monotonically increasing row counts in the queue, this gets the same bound as the out-of-order bound $k$ proved in Theorem \ref{thm:law-queue-bound}, and follows by essentially the same argument. This also holds for the static 2D queue.

%% file: Sections/Correctness/lpw-queue.tex
\subsection{LpW Elastic 2D Queue}

% Enumerating windows and nodes
This queues is not as trivial as the last one, as its \depth can change in both \Win{head} and \Win{tail} instantly, meaning that the \Win{head} will not necessarily follow the same sequence of row ranges as \Win{tail}. To simplify, we introduce an ordering of the windows with respect to their max, such that $\Win{i} < \Win{j}$ if $\Win{i}{max} < \Win{i}{max}$. We also denote the window during which item $x$ was enqueued or dequeued as \Win{enq x} or \Win{deq x} respectively. Finally, we denote the row of a node $x$ in a sub-queue or the \lateral queue as $\row{x}$.

\begin{lemma} \label{lemma:queue-passing}
    If a \lateral node $l$ is enqueued at time $t$, then $\Win{head}{max} < \row{l}$ at $t$.
\end{lemma}

\begin{proof}
We first note that $\row{l} \gets \Win{tail}{max} + 1$ when it is enqueued (line \ref{alg:lpw-queue-lat-max}). Furthermore, the enqueue of $l$ completes before the window shifts, and it cannot be enqueued again after the window has shifted as the comparison against the strictly increasing row count of the tail will fail (line \ref{alg:lpw-queue-lat-check}). As $\Win{head}{max} \leq \Win{tail}{max}$ (line \ref{alg:lpw-head-overtake}), and $\Win{tail}{max} < \row{l}$ at $t$, it holds that $\Win{head}{max} < \row{l}$ at $t$.
\end{proof}

% Rank error bound for 2D queue
% Now we still only relax it for the dequeues. We could add another theorem for relaxing both operations.
\begin{theorem}\label{thm:elastic-queue-bound}
    The elastic 2D LpW queue is linearizeable with respect to a FIFO queue with elastic $k$ out-of-order relaxed dequeues, where for every dequeue of $x$, $k = (\Win{enq x}{width} - 1)(\Win{enq x}{depth} + \Win{deq x}{depth} - 1)$.
\end{theorem}

\begin{proof}
    % Prove linearizability, and bound empty returns
    Enqueues and non-empty dequeues linearize with a successful update on a MS sub-queue. An empty dequeue linearizes when $\Win{head}{max} = \Win{tail}{max}$ after a double-collect where it sees all sub-queues within $\Win{head}{width}$ empty. Lemma \ref{lemma:queue-passing} together with $\Win{head}{max} = \Win{tail}{max}$ gives that all nodes must be within $\Win{head}{width}$, meaning that empty returns can linearize at a point where the queue was completely empty.
    
    % Bound k for successful returns
    The core observation for proving the out-of-order bound is to observe that the row counts for the sub-queues and the max of the windows strictly increase. For a node $y$ to be enqueued before $x$, and not dequeued before $x$, it must hold that $\Win{enq y} \leq \Win{enq x}$ and $\Win{deq y} \geq \Win{deq x}$. As $\Win{enq x}{min} \le \Win{deq x}{max}$, we can bound the possible row of $y$ by $\Win{deq x}{min} \le \row{y} \leq \Win{enq x}{max}$. Using Lemma \ref{lemma:queue-passing} together with the fact that \Win{head} only spans rows with one \width at a time (line \ref{alg:lpw-lateral-limit-max}), we get that these valid rows for $y$ must have width $\Win{enq x}{width} = \Win{deq x}{width}$. As both \Win{enq x} and \Win{deq x} must share at least $\row{x}$, it holds that $\Win{enq x}{min} \le \Win{deq x}{max}$. This is then used to limit the valid number of $\row{y}$ by $\Win{enq x}{max} - \Win{deq x}{min} = \Win{enq x}{min} + \Win{enq x}{depth} - \Win{deq x}{max} + \Win{deq x}{depth} \leq \Win{enq x}{depth} + \Win{deq x}{depth} - 1$. As all operations within each sub-queue are ordered, we reach the final bound of $(\Win{enq x}{width} - 1)(\Win{enq x}{depth} + \Win{deq x}{depth} - 1)$.
\end{proof}

% Lateness
Similarly to the LpW queue, this argument can be flipped to prove that $k$ also bounds the lateness or delay of the queue.

%% file: Sections/Correctness/lpw-stack.tex
\subsection{Elastic LpW 2D Stack}

% We will prove the bad bound
To bound the out-of-order error for an item in the elastic LpW 2D stack. This is done by first proving that the \lateral correctly bounds row widths in Lemma \ref{lemma:stack-invariant}, bounding the size of sub-stacks in Lemma \ref{lemma:substack-lower-bound} and \ref{lemma:substack-upper-bound}, and then deriving the out-of-order bound in Theorem \ref{thm:elastic-stack-bound}.

% Notation
As in the queue analysis, we use $\Win{push x}$ or $\Win{pop x}$, to denote the window when $x$ was pushed or popped. We also introduce \Win{max x}{depth} to denote the maximum value of  \Win{}{depth} during the lifetime of $x$, and equivalent notation for any window attribute. We denote $\Win{}{shift} = \floor{\Win{}{depth}/2}$, the top row (size) of sub-stack $j$ as $N_j$, and the index an item $x$ is pushed at as $\textit{index}_x$. Finally, we introduce a \textit{width bound} (\width{r}) for each row $r$ as $l.\width$ if there exists a \lateral node $l, l.next.row < r \leq l.row$, or \Win{}{push\_width} if $r > l.row ~ \forall ~ l$ (if this properly bounds the row widths, the \lateral is \textit{consistent}). Due to \lateral nodes being removed from the stack if their row overlaps the next node's row, this width bound is uniquely defined.

%% Stack invariant lemma
\begin{lemma}\label{lemma:stack-invariant}
    At the moment preceding the linearization of each window shift, it holds for each row $r$ and every item $x$ where$ \row{x} = r$, that $\width{r} \geq \textit{index}_x$.
\end{lemma}

\begin{proof}
    % Base case for induction
    Informally, this lemma states that the \lateral stack properly bounds all rows with widths greater than \Win{}{push\_width} after the call to \textit{Stabilize} at line \ref{alg:sync-lat-decl}. We prove this with induction over the sequence of window shifts, and it is easy to see that it holds at the first window shift, as all nodes will be pushed within \Win{}{put\_width}. Now, if the lemma held at the previous window shift, we check if it continues to hold for the next shift where we shift from \Win{i}.

    % Nodes below window
    First, we inspect rows at and below \Win{i}{min} and note that during the lifetime of \Win{i}, all nodes are pushed above \Win{i}{min}, and \width{r} will not be changed at rows at or below \Win{i}{min} (lines \ref{alg:lateral-update-base}, \ref{alg:lat-smaller-width-update}, \ref{alg:lat-larger-width-update}) from lowering a \lateral node. Thus, if $\Win{i}{push\_width} = \Win{i}{last\_push\_width}$, the lemma continues to hold for all rows lower or equal to \Win{i}{min}. If the $\Win{i}{push\_width} \neq \Win{i}{last\_push\_width}$ a new \lateral node can be inserted with width \Win{i}{last\_push\_width}. This node changes row bounds for rows at or below \Win{i}{min} if there was no other \lateral node above \Win{i}{min} at the shift to \Win{i}, and in that case those rows would have before \Win{i} been bounded by \Win{i}{last\_push\_width}, which is the same as the width bound this node enforces. Therefore, the induction invariant continues to hold for rows at or below \Win{i}{min}.

    % Nodes above bottom of window
    Now we inspect rows above \Win{i}{min} to see if the invariant also holds there. Firstly, no row will have a width bound smaller than \Win{}{push\_width}, as smaller widths will be lowered or inserted to or below \Win{i}{min} (lines \ref{alg:lat-smaller-width-update}, \ref{alg:lat-smaller-push-row}). Therefore, only items above \Win{i}{min} outside \Win{i}{push\_width} can break the invariant. In the lowering phase, \lateral nodes $l, l.\width > \Win{i}{last\_push\_width} \land l.\width > \Win{push\_width}$ are lowered iff the shift to \Win{i} was downwards, as then all sub-stacks outside \Win{i}{last\_push\_width} were seen at the bottom of the last window. Thus, if $\Win{i}{push\_width} \geq \Win{i}{last\_push\_width}$ no nodes could have been pushed outside \Win{i}{push\_width} since the shift to \Win{i}, and as all widths bound held then, they will hold at the shift from \Win{i}. Otherwise, if $\Win{i}{push\_width} < \Win{i}{last\_push\_width}$, every row $r, r \leq l.\row$ for the topmost \lateral node $l$ must have a valid bound, and the rows above $l.\row$ were before \Win{i} bounded by \Win{i}{last\_push\_width} which is smaller than the new bound \Win{i}{push\_width}, so the width bounds must hold for all rows in this case as well. %\kvgcom{These last conclusions are a bit informal, and pulls together many small sub-proofs}
\end{proof}

\begin{lemma} \label{lemma:substack-lower-bound}
    If $x$ lives on the stack during $'x$, then for any item $y$ pushed during $'x$, $\row{y} \ge \Win{push x}{min} - \Win{max x}{shift}$.
\end{lemma}

\begin{proof}
    This is proved by contradiction. Assume $x$ was pushed at time $t_x$ and there exists an item $y, \row{y} \leq \Win{x}{min} - \Win{max x}{shift}$, pushed at time $t_y < t_x$. Call the point in time where a thread shifted the window to \Win{y} as $t_s$. 
    \begin{itemize}
        \item If $t_s < t_x$, then $x$ must be pushed in the same, or a later window than \Win{y}. But it cannot be later as $y$ is pushed during \Win{y}, and $t_y > t_x$. Furthermore, as $t_x$ can't be during \Win{y} either, as an item is pushed at or below \Win{}{max}.
        \item If instead $t_s > t_x$, we call the window before \Win{y} as \Win{z}. For a thread to shift from \Win{z}, it must have seen $\forall j, N_j = \Win{z}{min}$ (as Lemma \ref{lemma:stack-invariant} shows iterating over \Win{}{pop\_width} is the same as iterating over all $j$) at some point $t_z$ (set $t_z$ as the time it started this iteration) during \Win{z}. As $\Win{z}{min} < \row{x}$, $t_x > t_z$, which means that $t_x$ must have been during \Win{z}, as we above showed $t_x < t_s$. This is impossible as during \Win{z}, items are pushed at or below \Win{z}{max}.
    \end{itemize}
\end{proof}

\begin{lemma} \label{lemma:substack-upper-bound}
    If $x$ lives on the stack during $'x$, then for any item $y$ pushed, but not popped, during $'x$, $\row{y} \leq \Win{pop x}{max} + \Win{max x}{shift}$.
\end{lemma}

\begin{proof}
    This is similar to the last lemma and is also proved by contradiction, assuming that $x$ was pushed at $t_x$, popped at time $t'_x$ and that there exists an item $y, \row{y} > \Win{pop x}{max} + \Win{max x}{shift}$ pushed at $t_y$ and not popped at $t'_x$, where $t_x < t_y < t'_x$. Call the point in time where a thread shifted to \Win{pop x} as $t_s$.
    \begin{itemize}
        \item If $t_s < t_y$, then $y$ must have been pushed in the same window as $x$ was popped. But items are only pushed at or below \Win{}{max}, which contradicts the assumption, as $y$ is pushed too low.
        \item If $t_y < t_s$, we call the window proceeding \Win{pop x} as \Win{z}. For a thread to shift from \Win{z}, it must have seen $\forall j, N_j = \Win{z}{min}$ at some point $t_z$ (set $t_z$ as the time it started the iteration, seeing the first $N_j = \Win{z}{min}$) during \Win{z}. Therefore, $t_z < t'_y$, which is impossible as $\Win{z}{max} < \row{y}$ and $t_y < t_s$.
    \end{itemize}
\end{proof}

% Rank error bound for elastic 2D stack
\begin{theorem} \label{thm:elastic-stack-bound}
    The elastic 2D stack is linearizeable with respect to a stack with $k$ out-of-order relaxed non-empty pops, where $k$ is bounded for every item $x$ as $k = (\Win{max x}{width} - 1) (3\Win{max x}{depth} - 1)$.
\end{theorem}

\begin{proof}
    Assume that $y$ is pushed after $x$ and not popped before $x$. Then Lemma \ref{lemma:substack-lower-bound} and \ref{lemma:substack-upper-bound} gives $\Win{push x}{min} - \Win{max x}{shift} \le \row{y} \leq \Win{pop x}{max} + \Win{max x}{shift}$. As each sub-stack is internally ordered and the maximum number of items pushed after $x$ and not popped before $x$ becomes $(\Win{max x}{push\_width} - 1)(\Win{pop x}{max} - \Win{push x}{min} + 2\Win{max x}{shift}) \leq (\Win{max x}{push\_width} - 1)(3\Win{max x}{depth} - 1)$.
\end{proof}

%% file: Sections/Evaluation/index.tex
\section{Evaluation}\label{sec:evaluation}

% Purpose and machine used
We experimentally evaluate the scalability and elastic capabilities of our elastically relaxed LaW queue, LpW queue, and LpW stack. All experiments run on an AMD-based machine with a 128-core 2.25GHz AMD EPYC 9754 with two-way hyperthreading, 256 MB L3 cache, and 755 GB of RAM. The machine runs openSUSE Tumbleweed 20240303 which uses the 7.4.1 Linux kernel. All experiments are written in C and are compiled with gcc 13.2.1 and its highest optimization level. In all tests, threads are pinned in a round-robin fashion between cores, starting to use hyperthreading after 128 threads.

% Our implementations of elastic data structures, and SSMEM
Our elastic 2D implementations build on an optimized version of the earlier 2D framework\,\cite{2D}. We use SSMEM from \cite{ASCYLIB} for memory management, which includes an epoch-based garbage collector for our dynamic memory. \kvgcom{We should publish the code at some point}

\input{Sections/Evaluation/static}

\input{Sections/Evaluation/manual}
\input{Sections/Evaluation/dynamic}

%% file: Sections/Evaluation/static.tex
\subsection{Static Relaxation}

% Goal, and what data structures selected
To understand the performance of our data structures under static relaxation, we compare their scalability against that of well-known relaxed and strict designs. For the queues, we select the static 2D queue\,\cite{2D} and the k-segment queue\,\cite{k-queue} as baselines for modern k-out-of-order designs. We selected the wait-free (WF) queue from \cite{wfqueue} as the state-of-the-art linearizable FIFO queue for our study. Although we also evaluated the LCRQ from \cite{lcrq}, it showed slightly inferior performance in our tests. Additionally, we included the Michael-Scott (MS) queue \cite{ms-queue} as a baseline.  For the stacks we similarly selected the 2D stack\,\cite{2D} and the k-segment stack\,\cite{quantitative-relaxation} for relaxed designs. We then also compared with a lock-free elimination-backoff stack\,\cite{elimination-backoff-stack} and the Treiber stack as a baseline\,\cite{treiber-stack}. All data structures were implemented in our framework using SSMEM \cite{ASCYLIB} for memory management, with the exception of the WF queue, for which we used the authors' implementation that employs hazard pointers.

% Test setup
We use a benchmark where threads over $1$ second repetedly perform insert or remove operations at random every iteration. Each data structure is pre-filled with $2^{19}$ items to avoid empty returns, which significantly alter the performance profile. Test results are aggregated over 10 runs, with standard deviation included in the plots. The test bounds the rank error of the data structures, and as the optimal choice of \Win{}{width} and \Win{}{depth} is not known\,\cite{2D}, we simply set $\Win{}{width} = 2\times\textit{nbr\_threads}$, and use the maximum \depth to stay within the bound, which is simple and gives acceptable scalability.

% Relaxation measurements
Measuring rank errors without altering their distribution is an open problem, and we adapt the method used, e.g. \cite{2D,spraylist,multiqueues-short} to our algorithms. It encapsulates the linearization points of all methods by global locks, imposing a total order on all operations, and keeps a sequentially ordered data structure on the side. After each removal operation that returns $x$, the distance between $x$ and the top item gives the rank error. This strategy greatly reduces throughput, as it serializes all operations, so all rank error measurements are done in separate runs from throughput runs. %We also contemplated timing every linearization point and using timestamps to approximate the linearization order, but the overhead of the timestamps was also significant compared to the operation time for every operation.

% Scaling with threads
Figure \ref{fig:static-scalability} shows how the queues and stacks scale with both threads and relaxation. The results show that the elastic designs scale essentially as well as the static 2D framework and out-scale the other data structures. This means that the overhead induced by \lateral is minimal in periods of static relaxation. This is due to the extra work from the \lateral being mostly confined to small checks during window shifts, and otherwise functions as the static 2D structures.

%%% SCALING STATICALLY %%%
\begin{figure}[!ht]
    \centering
    \includegraphics[width=0.45\textwidth]{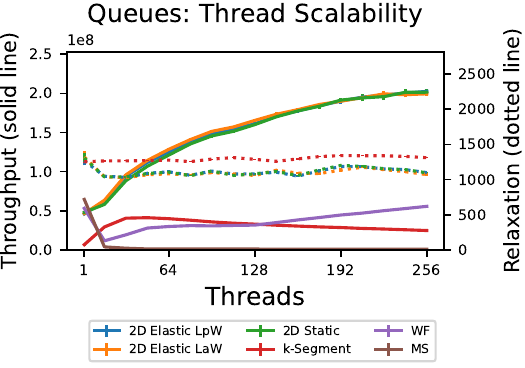}
    \hspace{0.05\textwidth}
    \includegraphics[width=0.45\textwidth]{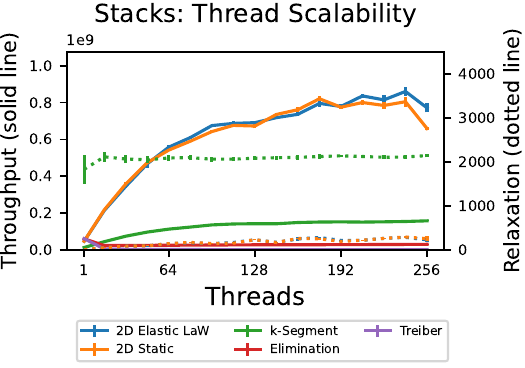}

    \vspace{0.5em}

    \includegraphics[width=0.45\textwidth]{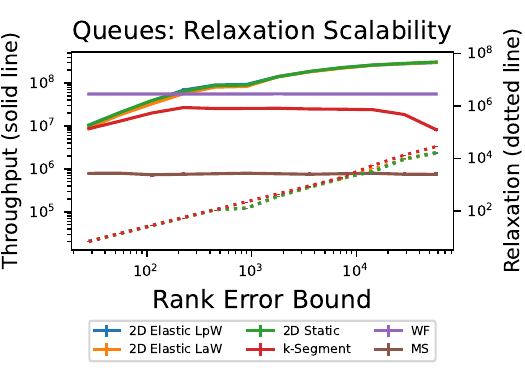}
    \hspace{0.05\textwidth}
    \includegraphics[width=0.45\textwidth]{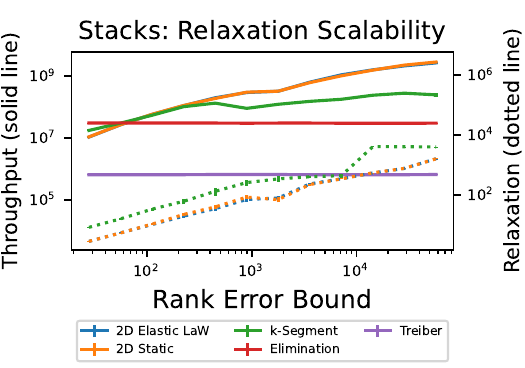}

    \caption{Scalability of throughput and rank error during static relaxation. When scaling with threads (top row), the error bound is fixed as $k=5\times10^3$. When scaling with error bound, 256 threads are used.}
    \label{fig:static-scalability}
\end{figure}

%% file: Sections/Evaluation/manual.tex
\subsection{Elastic Relaxation - Manual adjustments}

% Intro and test setup
Here, we demonstrate that our elastic designs can efficiently trade accuracy for throughput by showing real-time relaxation and throughput metrics for a benchmark with several elastic changes in relaxation. Figure \ref{fig:elastic-changes-manual} plots throughput as well as average rank errors for a single run for our different designs. Vertical lines represent user-initiated elastic relaxation adjustments. Our data structures are pre-filled with $2^{15}$ items, as to avoid empty returns while keeping the number of items relatively small. The threads all save a timestamp every 1000 operations to measure the throughput and the throughput values are then aggregated with a moving average window of size $25$~ms. All 128 threads randomly alternate between inserting and removing items.

% How we measure relaxation, scaling time axis
Measuring the rank errors is done in a separate run from the throughput, as the measurements impose a total ordering on all operations, massively slowing down the throughput. To align the relaxation and throughput measurements, we ran the relaxation test for $10^3$ times as long for the queue ($2\cdot10^3$ for the stack) and subsequently compressed the relaxation measurements to the same time frame. These factors were used because they were the average slowdown of the throughput when measuring relaxation during the experiment. The rank errors are also smoothed out with the same moving average window of $25$~ms.

%%% ELASTIC CHANGES %%%
\begin{figure}[!ht]
    \centering
    \includegraphics[width=0.49\textwidth]{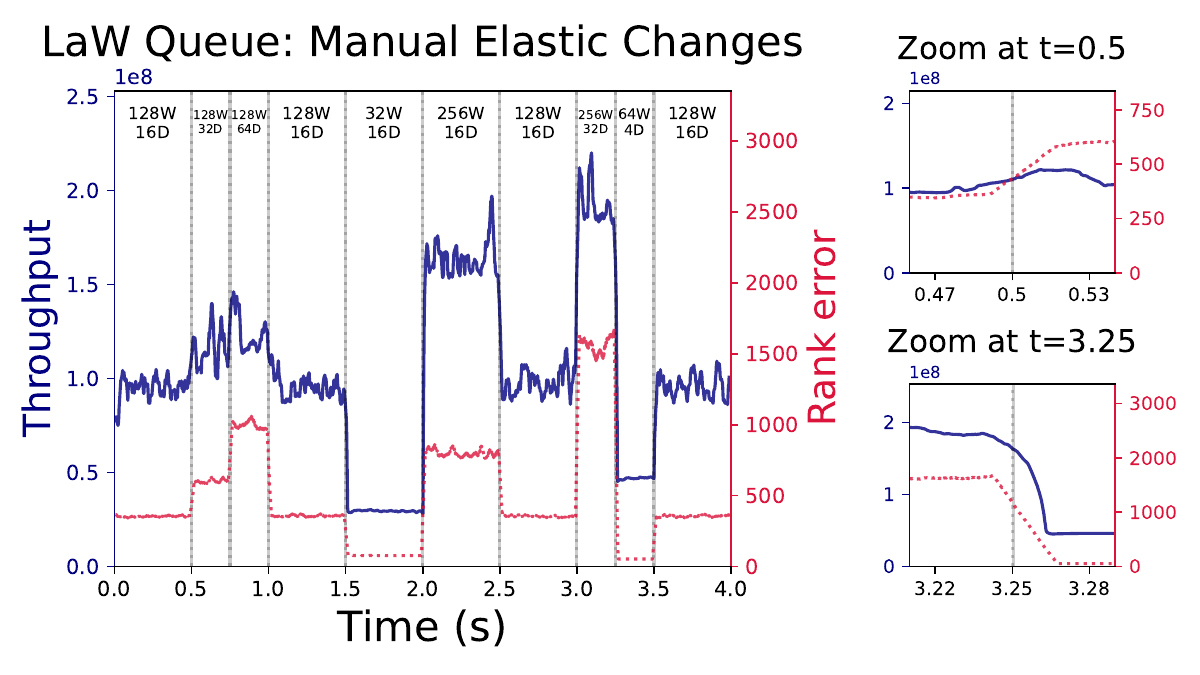}
    \hfill
    \includegraphics[width=0.49\textwidth]{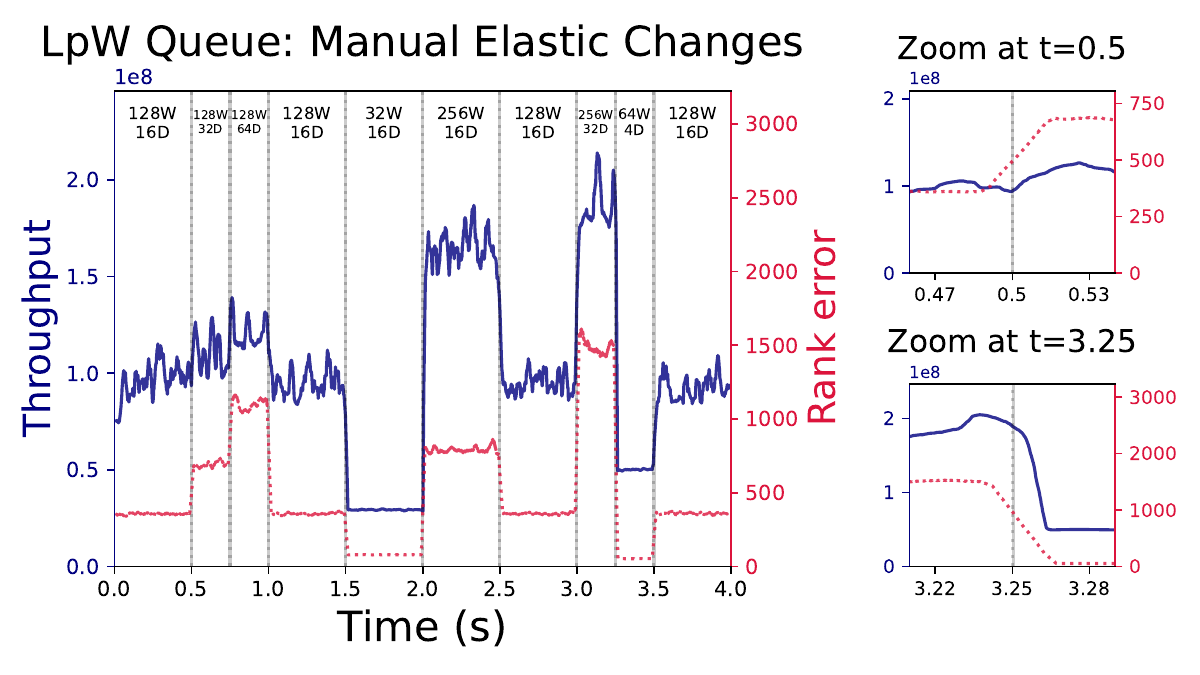}
    \hfill
    \includegraphics[width=0.49\textwidth]{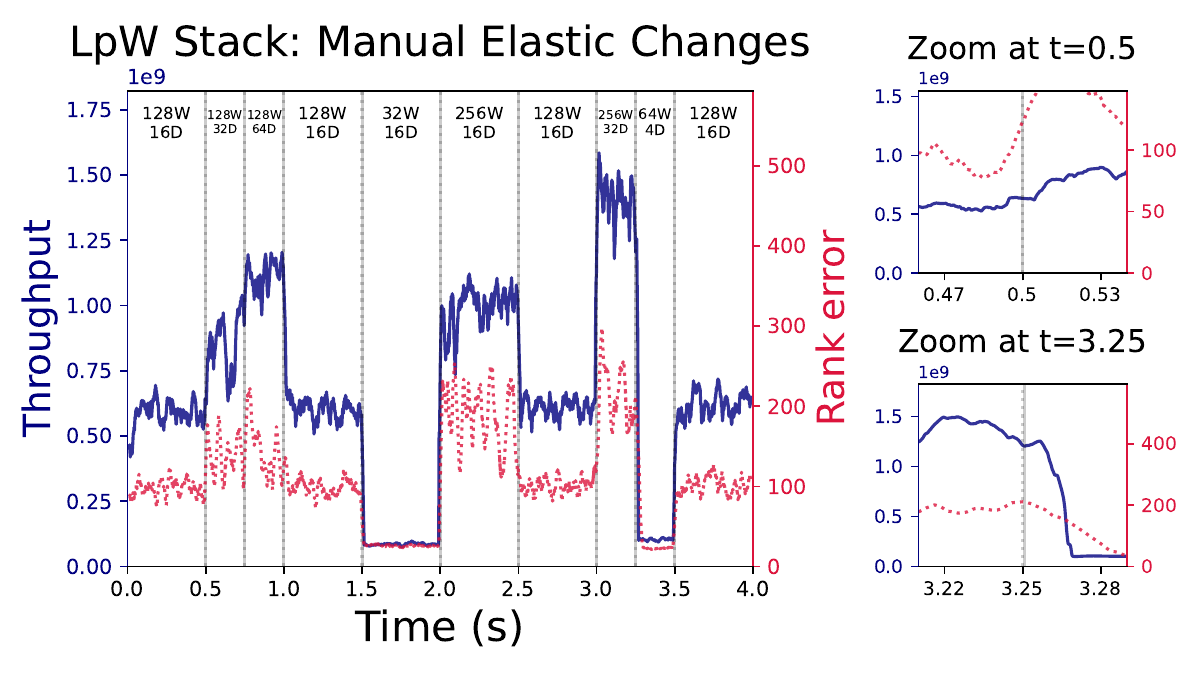}
    \caption{Performance during one run, where threads repeatedly insert or remove items at random, and vertical lines indicate user-initiated elastic changes. It runs for four seconds with 128 threads. The width (W) and depth (D) are annotated for every relaxation period.}
    \label{fig:elastic-changes-manual}
\end{figure}

% Interpreting the graphs
As show in Figure \ref{fig:elastic-changes-manual}, the elastic relaxation works and can quickly change throughput and rank errors. For the queues, the rank error measurements are relatively stable, and so is the throughput, althought to a lesser extent. The throughput is dependent on randomness and how the threads select sub-queues to operate on, which leads to variance in the performance over time. The stack has similar variance in throughput, but more erratic average rank errors. The rank error bound for the stack does not only depend on the current window, but also previous ones, which means it can sometimes change unexpectedly. The stack also has a much more erratic performance profile, as the threads can work completely thread-local on a sub-stack for a long time in the best case with close to 0 rank error, but in the worst case can have to jump around a lot.  

% Conclusions - a bit of a tie in with applications.
Although requiring externally initiated changes in relaxation, the plots show how swiftly the designs can trade accracy for throughput. It is evident that the optimial choice of \width and \depth depends on the data structure, as well as the workload, meaning these plots can be of use when creating applications where you want to control the relaxation differently during different time periods. 

%% file: Sections/Evaluation/dynamic.tex
\subsection{Elastic Relaxation - Dynamic adjustments}

% Introduce the idea, and the test setup
Here we showcase the elastic capabilities of our designs by implementing a simple controller that dynamically controls \Win{}{width} in a realistic use case with dynamic workload. Consider a shared queue where tasks can be added and removed in FIFO order. We let one third (42) of the 128 cores act as consumers of this queue, constantly trying to dequeue tasks from it. The remaining (86) cores are designated as producers, repeatedly adding tasks to the queue, though not always active, as illustrated in the top graphs of Figure \ref{fig:variable-workload}. This simulates a task queue in a highly contended server, where the consumers are internal workers working at a constant pace, and the load of the producers vary depending on external factors. Our goal is to control the relaxation to cope with the dynamic nature of this producer workload.

% Introduce the controller
To cope with this dynamic nature, our relaxation controller strives to keep the operational latency approximately constant for all producers. To minimize overhead, this controller is thread-local and only tracks failed and successful CAS linearizations, and is shown in Algorithm \ref{alg:elastic-queue-controller}. It increments or decrements a \textit{contention} count by \textit{SUCC\_INC} or \textit{FAIL\_DEC} depending on if the CAS linearization on a sub-queue succeeds or fails respectively (line \ref{alg:controller-inc}). If $|\textit{contention}| > \textit{CONT\_TRESHOLD}$, the thread resets \textit{contention}, adds a local vote for increasing or decreasing the next \Win{tail}{width} by \textit{WIDTH\_DIFF}, and depending on its votes tries to change \textit{width\_shared} (used at lines \ref{alg:law-queue-width} and \ref{alg:lpw-queue-width-shared}). When the window shifts, the local vote count is reset. These values can be tuned, but are in our experiments set as $\textit{SUCC\_INC} = 1, \textit{FAIL\_DEC} = 75, \textit{CONT\_TRESHOLD} = 5000, \textit{WIDTH\_DIFF} = 5$.

% Pseudo code for the windowless lateral
\begin{algorithm}[!ht]
\SetKwProg{Func}{function}{}{}
\SetKwProg{Method}{method}{}{}
\SetKwProg{Struct}{struct}{}{}
\SetKwComment{NormComment}{// }{}

\caption{Simple relaxation controller for the Elastic 2D queues}\label{alg:elastic-queue-controller}
\scriptsize
%\footnotesize
%\DontPrintSemicolon
%\SetKwComment{Comment}{$\triangleright$\,}{}

\Struct{Controller}{
    uint contention\; 
    uint version\;
    uint votes\;
}

\Method{Controller.Update(self, cas\_success) \label{alg:controller-update}} {
    \If{self.version $\neq$ \Win{tail}{max}}
    {
        self.version $\gets$ \Win{tail}{max}\;
        self.votes $\gets$ 0\;
    }

    self.contention $\gets$ self.contention + \textbf{if} cas\_success \textbf{then} SUCC\_INC \textbf{else} -FAIL\_DEC\; \label{alg:controller-inc}

    \If{$|$self.contention$|$ $\geq$ CONT\_TRESHOLD}
    {
        self.votes $\gets$ self.votes + sign(self.contention)\;
        self.contention $\gets$ 0\;
        width$_\textit{shared}$ $\gets$ \Win{tail}{width} - WIDTH\_DIFF $*$ sign(self.votes)\;
    }
}
\end{algorithm}

% Describe the plots
Figure \ref{fig:variable-workload} shows the average thread operational latency, as well as the error bound $(\Win{head}{width} - 1)\times\Win{head}{depth}$ averaged over 50 runs for our elastic LaW queue. It also shows $\textit{tail error} = (\Win{tail}{width} - 1)\times\Win{tail}{depth}$ which can be interpreted as a rank error for enqueues, and shows that the controller adapts quickly. However, it notes occasional delays between \Win{tail}{width} and \Win{head}{width}, as the change has to propagate through the queue. The figure shows a scenario without the dynamic controller, one short test over $1$ second, and one long test over $1$ minute.

%%% ELASTIC CONTROLLER PLOTS %%%
\begin{figure}[!ht]
    \centering
    
    \begin{subfigure}[b]{0.45\textwidth}
        \centering
        \includegraphics[width=\textwidth]{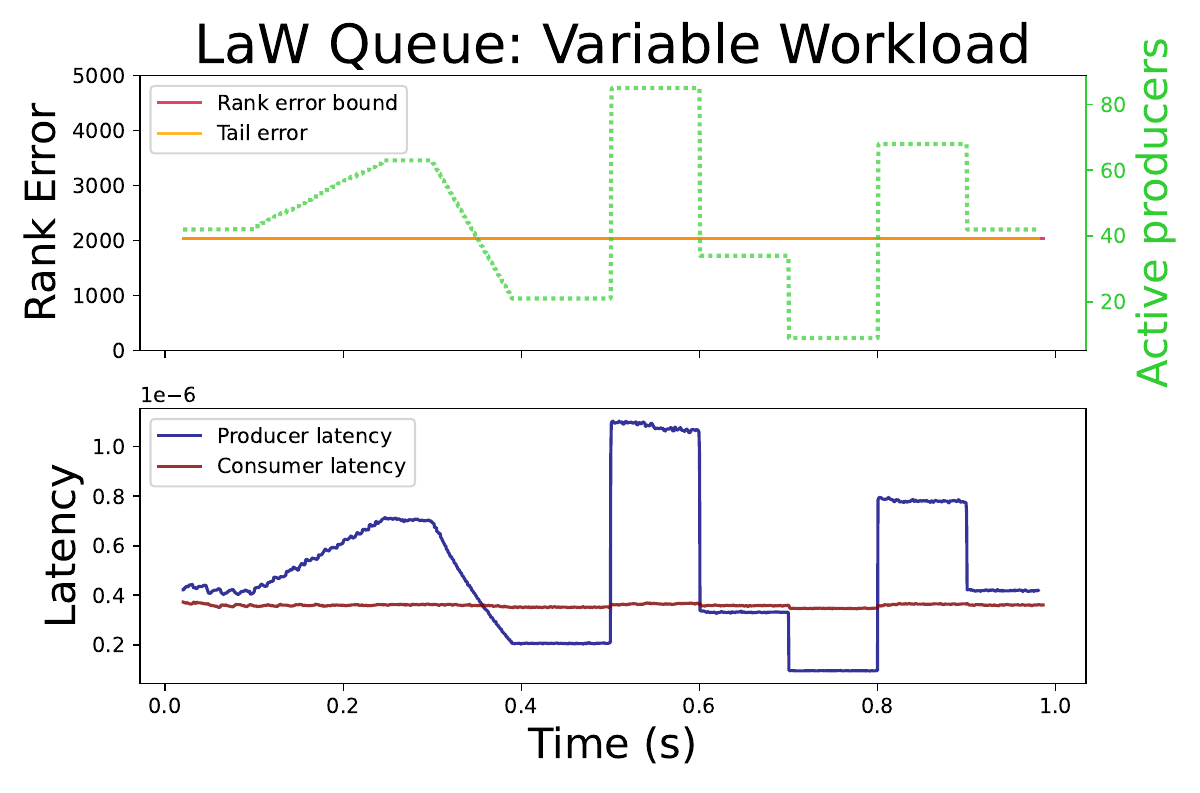}
        \caption{Static Relaxation}
        \label{fig:uncontrolled}
    \end{subfigure}
    \hfill % optional, for aesthetics
    \begin{subfigure}[b]{0.45\textwidth}
        \centering
        \includegraphics[width=\textwidth]{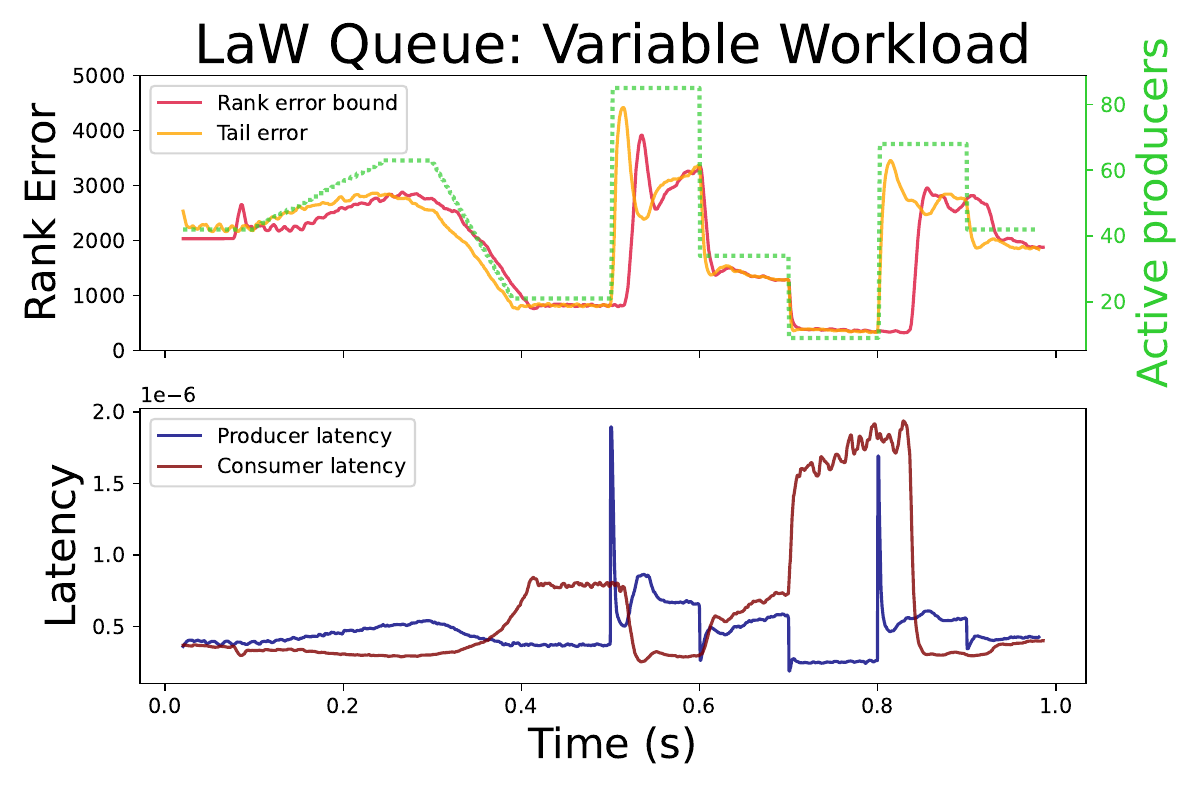}
        \caption{Dynamic Relaxation, 1 second}
        \label{fig:controlled_1s}
    \end{subfigure}
    \hfill % optional, for aesthetics
    \begin{subfigure}[b]{0.45\textwidth}
        \centering
        \includegraphics[width=\textwidth]{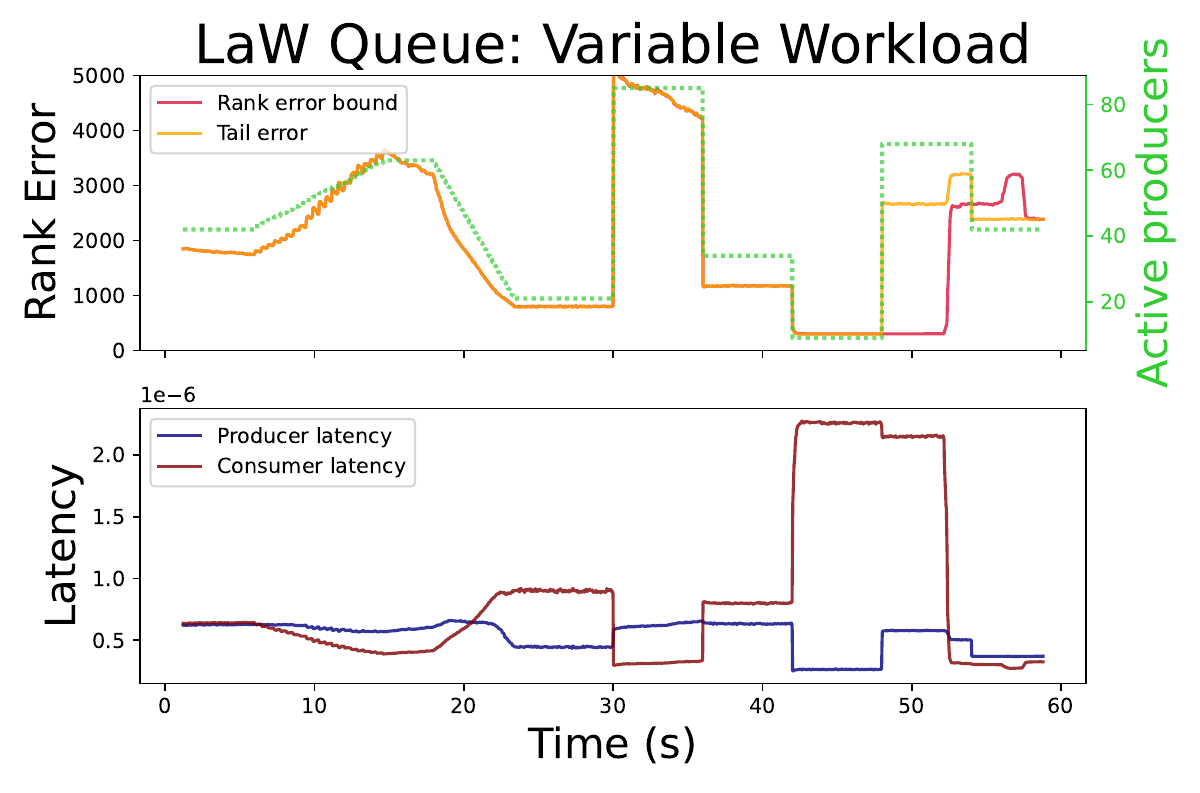}
        \caption{Dynamic Relaxation, 1 minute}
        \label{fig:controlled_1m}
    \end{subfigure}

    \caption{Producer-consumer system with a variable number of producers over time for the elastic LaW queue. The \textit{Dynamic Relaxation} plots use a controller to adjust the relaxation, stabilizing the producer's latency.}
    \label{fig:variable-workload}
\end{figure}

\kvgcom{Could probably shorten this a bit}

% Interpret the plots
The test without the controller shows how the latency clearly scales with contention for the producers, while the consumers are mostly unaffected. However, using the controller, the producers' latency is much more stable. While latency spikes are evident with rapid increases in contention, the controller effectively stabilizes them. While producers enjoy stable latency, consumers must accommodate the significant variations in relaxation. Furthermore, it is evident that the controller quickly adjusts \Win{tail}{width} based on the number of producers. However, at some points, it takes a while for this change to reach \Win{head} and affect the rank error bound, as it must propagate through the whole queue.

% Problems with controller, conclusion
This experiment shows a practical use-case for the elastic queue, and that a simple thread-local controller for the \width is enough to get good dynamic trade-offs between relaxation and latency. Similar controllers could easily be created to target different use-cases, such as those where we care about the performance of both the producers and consumers. For example, using the elastic LaW queue, a controller could control \Win{head}{depth} and \Win{tail}{depth} separately in combination with \Win{}{width}, which would lead to more flexible adaptation. To fully leverage such a controller, it would be helpful to design a model for the queue performance, so that the choices of \depth and \width could be made with more information.

%% file: Sections/Conclusion.tex
\section{Conclusions}\label{sec:conclusion}

We have presented the concept of elastic relaxation for concurrent data structures, and extended the 2D relaxed framework from \cite{2D} to encompass elasticity. The history of elastic changes is tracked by the \lateral structure, which can also be used to extend other $k$ out-of-order data structures. Our designs have established worst-case bounds, and demonstrate as good performance during periods of constant relaxation as state-of-the-art designs, while also being able to reconfigure their relaxation on the fly. Our simple controller, based on thread-local contention, demonstrated that the elasticity can be utilized to effectively trade relaxation for latency. 
We believe elastic relaxation is essential for relaxed data structures to become realistically viable and see this paper as a first step in that direction.

As further work, we find constructing a model over the data structure performance interesting, which could aid in designing more sophisticated relaxation controllers. Another direction is applying the idea of elasticity to other data structures, such as relaxed priority queues.

%% file: Sections/Appendix/2Dc-stack-proof.tex
\section{New bound for static 2D stack} \label{app:2Dc-stack-bound}

\comment{
    \newcommand{\sDDc}{2Dc-Stack\xspace}
    \newcommand{\koo}{k-out-of-order\xspace}
    \newcommand{\pusho}{Push\xspace}
    \newcommand{\popo}{Pop\xspace}
    \newcommand{\substak}{sub-stack\xspace}
    \newcommand{\substaks}{sub-stacks\xspace}
    \newcommand{\window}{Window\xspace}
    \newcommand{\Wmax}{W^{max}\xspace}
    \newcommand{\Wmin}{W^{min}\xspace}
}   

The relaxation bound for the static 2D stack from \cite{2D} (Theorem 5) has a small mistake, and is too low. The error comes from the fact that the proof mistakenly assumes that each item is popped while \Win{}{max} is the same as when it was pushed. If $\depth=2\shift$, the difference between these two \Win{}{max} is \shift. Here we present a new bound and corresponding proof. 

For brevity, we call the number of nodes on sub-stack $j$ as $N_j$, and the window item $x$ was popped $\Win{pop x}$. We also note that \shift does not necessarily have to be $\depth/2$ which we assume in the rest of the paper, as this is not enforced for the proofs in \cite{2D}.

\begin{lemma} \label{lemma:substack-onesided-bound}
    It always holds that $\forall j: \Win{min} - \shift \leq N_j \leq \Win{}{max}$ or $\forall j: \Win{min} \leq N_j \leq \Win{}{max} + \shift$.
\end{lemma}

\begin{proof}
This is equivalent to the already proven Lemma 4 in \cite{2D}.
\end{proof}

\begin{theorem} \label{thm:original_six}

The static 2D stack has a $k$ out-of-order relaxation bound of \\
$k=\left(2\depth + 2\shift + \floor{\frac{\depth - 1}{\shift}}\shift\right)(\width-1)$.

\end{theorem}

\begin{proof}

    This is proved by looking at how many items can be pushed during the lifetime of $x$ ($'x$), without being popped during $'x$. Lemma \ref{lemma:substack-onesided-bound} bounds the size of all sub-stacks when $x$ is popped as $N_j \leq \Win{pop x}{max} + \shift = N_\textit{upper}$.

    Additionally, Lemma \ref{lemma:substack-onesided-bound} bounds the size of all sub-stacks during $'x$ as $N_j \geq \row{x} - \shift - \depth = N_{\textit{lower}}$. Furthermore, the lemma limits $N_{\textit{lower}}$ to be at $\Win{}{min}$, or $\shift$ below $\Win{}{min}$ for some \Win{}{min} during $'x$, meaning (i) $\Win{pop x}{max} - N_{\textit{lower}} = \depth + m\times\shift$ for some integer $m$.

    All items in the same sub-stack as $x$ are correctly ordered. Therefore, the number of items pushed and not popped during $'x$ must be at most $(\width - 1)(N_\textit{upper} - N_\textit{lower}) = (\width - 1)(\Win{pop x}{max} - \row{x} + 2\shift + \depth)$. From (i) we have that $\Win{pop x}{max} - \row{x}$ is an even multiple of \shift. Additionally, as $x$ is popped during $\Win{pop x}$, it holds that $\Win{pop x}{max} - \row{x} \leq \depth - 1$. Combining these two last statements gives that $\Win{pop x}{max} - \row{x} \leq \floor{\frac{depth - 1}{shift}}\shift$, which bounds the number of items pushed, but not popped, during $'x$ as $(\width - 1)(2\shift + \depth + \floor{\frac{depth - 1}{shift}}\shift)$.
 
\end{proof}